\documentclass[12pt]{article}

\usepackage{amsmath}
\usepackage{amssymb}
\usepackage{amsfonts}
\usepackage{cite}
\usepackage{graphicx}
\usepackage{subfigure}
\usepackage{tikz}
\usepackage{hyperref}
\usepackage{fullpage}
\usepackage{wasysym}

\usepackage{setspace}

\date{}

\newtheorem{theorem}{Theorem}[section]
\newtheorem{lemma}{Lemma}[section]

\newenvironment{proof}[1][Proof]{\begin{trivlist}
\item[\hskip \labelsep {\bfseries #1}]}{\end{trivlist}}

\title{Distributed boundary tracking using alpha and Delaunay-\v{C}ech shapes}
\author{Harish Chintakunta\\
\small{North Carolina State University}\\
\small{hkchinta@ncsu.edu}
\and
Hamid Krim\\
\small{North Carolina State University}\\
\small{ahk@ncsu.edu}
}

\begin{document}
\maketitle

\begin{abstract}
For a given point set $S$ in a plane, we develop a distributed algorithm to compute the $\alpha-$shape of $S$. $\alpha-$shapes are well known geometric objects which generalize the idea of a convex hull, and provide a good definition for the shape of $S$. We assume that the distances between pairs of points which are closer than a certain distance $r>0$ are provided, and we show constructively that this information is sufficient to compute the alpha shapes for a range of parameters, where the range depends on $r$.

Such distributed algorithms are very useful in domains such as sensor networks, where each point represents a sensing node, the location of which is not necessarily known.

We also introduce a new geometric object called the Delaunay-\v{C}ech shape, which is geometrically more appropriate than an $\alpha-$shape for some cases, and show that it is topologically equivalent to $\alpha-$shapes.

\end{abstract}

\section{Introduction}

Many applications call for detecting and tracking the boundary of a dynamically changing space of interest \cite{5703088}\cite{chintakuntaCAMSAP2011}. We would expect any algorithm performing the task to include the following important properties: 1) the boundary output is geometrically close to the actual boundary,  and 2) the interior of the boundary is topologically faithful to the original space. It is often the case that we are only given random samples from the space. We may then reconstruct the space by first placing balls of a certain radius around these points, and then by taking the union of these balls. A good exposition on relationship amongst the sampling density, the geometry of the underlying space, and the radius of the balls may be found in \cite{niyogiSampling2009}. \\

In this paper, we start with the assumption that the union of the balls described above is a good approximation to the space of interest. Note that in some cases, this is by design. For example, in the case of systematic failures in sensor networks \cite{chintakuntaCAMSAP2011}, the failure in the nodes is caused by a spatially propagating phenomenon, and our aim is to track its boundary. In this case, we construct a space by taking the union of balls of radius $r_c/2$ around each node, where $r_c$ is its radius of communication. The radius of communication is the distance within which two nodes can communicate with each other. \\

The problem may also be viewed as one of computing the boundary of a set of points, provided with some geometric information. Given the binary information, about nodes pair-wise location within a certain distance, distributive algorithms exist to compute such a boundary \cite{5508237}. These algorithms are, unfortunately relatively slow, on account of the need for a global structure to reach a decision on the membership of a node or of an edge to the complex boundary. If on the other hand, we are provided with all pair-wise distances of nodes within a neighborhood, the above decision may be locally made by constructing an associated  $\alpha-$shape. \\

Given a set of points $S$ in a plane, the $\alpha-$shape introduced in \cite{edelsbrunner1983shape} gives a generalization of the convex hull of $S$, and an intuitive definition for the shape of points. More importantly, an $\alpha-$shape is the boundary of an alpha complex, which has the same topology as that of the union of balls. This relation amongst $\alpha-$shape, alpha complex and the union of balls is contingent on certain relations between their parameters. We discuss this in detail in Section \ref{sec: alpha-complex and alpha-shape}. Such topological guarantees cannot be provided by the boundary computed in \cite{5508237}. \\

The Delaunay triangulation gives sufficient information to compute the alpha complex, and hence its boundary, the $\alpha-$shape. If we only require the $\alpha-$shape, less information would be necessary.  The work in \cite{fayed2009localised} shows that a global Delaunay triangulation is not necessary, and the alpha shape can be computed using local Delaunay triangulations. Given the edge lengths of a geometric graph, \cite{fayed2009localised}  constructs the local Delaunay triangulation by first building local coordinates and then computing the Delaunay triangulation. Computing the local coordinates,  is however not robust and requires a high density of nodes for accuracy. 
When given the edge length information, we show that even local Delaunay triangulation is not necessary. \\

When there is a sufficient density of nodes, computing local coordinates is accurate (probabilistically), and distributed algorithms exist for computing modified versions of Delaunay triangulation \cite{ChenAvinThesis,li2003localized}. In this case, we define a certain \emph{Delaunay-\v{C}ech triangulation},  which contains an alpha complex, and which we show to be homotopy equivalent. For boundary tracking-based applications, the boundary of Delaunay-\v{C}ech triangulation will serve as a better geometric approximation to the boundary, while preserving the topological features. \\


Our contributions in this article are:
\begin{itemize}
\item Given the distances between  pairs of nodes whenever they are closer than $r_c>0$, we develop an algorithm to compute the $\alpha-$shape for a range of parameters, where this range depends on $r_c$.
\item We introduce the Delaunay-\v{C}ech triangulation, defined in Section \ref{sec:Restricted Delaunay Triangulation}, and show that it is homotopy equivalent to the alpha complex.
\end{itemize}

The remainder of the paper is organized as follows, in Section \ref{sec:Preliminaries}, we provide some background information, along with a formulation of the problem. We describe the distributed algorithm for computing an $\alpha-$shape in Section \ref{sec:computingAlphaShape}. The Delaunay-\v{C}ech triangulation and the Delaunay-\v{C}ech shape are defined in Section \ref{sec:Restricted Delaunay Triangulation}, while the proof of its  topological equivalence to the alpha complex is given in Section \ref{Relation between RDT(G) and  alpha- complex}. We conclude in Section \ref{sec:Conclusion} with some remarks.

\section{Preliminaries}
\label{sec:Preliminaries}

\subsection{Alpha complex and $\alpha-$shape}
\label{sec: alpha-complex and alpha-shape}
Consider a set of nodes  $V \subset \mathbb{R}^2$, and a parameter $r$. Let $V_i$ be the voronoi cell associated with node $v_i \in V$ in the voronoi decomposition of $V$.  Define an alpha cell ($\alpha-$cell) of $v_i$ as $\alpha(v_i,r) = V_i \cap B(v_i,r)$ where $B(v_i,r/2)$ is the closed ball of radius $r/2$ around $v_i$. The alpha complex, $A_r$ (we are assuming $V$ is implied in this notation), is defined as the nerve complex of the alpha cells, i.e.,  $(v_0,v_1,\ldots,v_k)$ spans a $k-$simplex in $A_r$ if $\bigcap_i \alpha(v_i) \ne \varnothing $. Since the alpha cells are convex, the nerve theorem \cite{bott1982lw,leray1945forme}  implies  that the alpha complex has the same homotopy type as the union of the alpha cells, which in turn is equal to the union of the balls $B(v_i,r/2)$. \\

Given a set of nodes $V \subset \mathbb{R}^2$ \footnote[2]{The alpha shape is generally defined  for points in $\mathbb{R}^k$ for any dimension $k$.}, and a parameter $r>0$, the alpha shape, $\partial A_r$, is a 1-dimensional complex which generalizes the convex hull of $V$. To simplify the notation, we use $(v_i,v_j)$ to denote an edge in a graph, a 1-simplex in a complex or the underlying line segment. A 1-simplex $(v_i,v_j)$  belongs to $\partial A_r$ if and only if a circle of radius $r/2$ passing through $v_i$ and $v_j$ does not contain any other node inside it. By ``inside'' a circle, we mean the interior of the ball to which this circle is a boundary.  We say that such a circle satisfies the ``\emph{$\alpha-$condition}''. $\partial A_r$ also contains all the nodes $\{v_j\}$ such that a circle of radius $r$ passing through  $v_j$ satisfies the $\alpha-$condition.   \\

For a 2-dimensional simplicial complex $K$, we define the boundary of $K$ to be the union of all the $1$-simplices (along with their faces), where each is a face of at most one $2-$simplex, and all $0-$simplices which are not faces of any simplex in $K$.  The alpha shape $\partial A_r$ is the boundary of the alpha complex $A_r$\cite{edelsbrunner1994three}.

\subsection{Delaunay-\v{C}ech Shape}
\label{sec:Restricted Delaunay Triangulation}
For a set of nodes $V \subset \mathbb{R}^2$ and a parameter $r>0$, define the geometric graph $G_{r} = (V,E)$ to be the set of vertices $(V)$ and edges $(E)$, where $e=(v_i,v_j)$ is in $E$ if the distance between $v_i$ and $v_j$ is less than or equal to $r$. Let $\check{C}(V,r)$ denote the \v{C}ech complex with parameter $r$ (the nerve complex of the set of  balls $\{B(v_i,r/2)\}$) and let $DT(V)$ be the Delaunay triangulation of $V$. We define the Delaunay-\v{C}ech complex $D\check{C}_r$ with parameter $r$ as $D\check{C}_r = DT(V)\cap\check{C}(V,r)$.  We will show in Section \ref{Relation between RDT(G) and  alpha- complex}, that $D\check{C}_r$ is homotopy equivalent to $A_r$. We call the boundary of $D\check{C}_r$, denoted by $\partial D\check{C}_r$ the Delaunay-\v{C}ech shape.

\section{Computing the alpha shape of points in $\mathbb{R}^2$}
\label{sec:computingAlphaShape}
In order to compute $\partial A_r$, we take each edge in $G_r$ and check if it is in $\partial A_r$. If an edge $e=(v_i,v_j)$ belongs to $\partial A_r$, then the length of the line segment $(v_i,v_j)$ is less than or equal to $r$. Otherwise, the $\alpha-$condition cannot be satisfied. Edge $e$ hence also belongs to $G_r$, and consequently, checking for all the edges in $G_r$ is sufficient to compute $\partial A_r$.\\

Given an edge $e=(v_i,v_j)$,  there are two circles of radius $r$ passing through $v_i$ and $v_j$. Let us call these circles $\mathcal{C}$ and $\mathcal{C}'$ (see Figure \ref{fig:cycleConstruction}). The $\alpha-$ condition is satisfied if and only if at most one of $\mathcal{C}$ and $\mathcal{C}'$ contains node(s) inside. \\

We consider all the nodes in $\mathcal{N}_i\cap\mathcal{N}_j$ (neighbors common to both $v_i$ and $v_j$) , and perform a series of tests to verify their location inside $\mathcal{C}$ and $\mathcal{C}'$. It is simple to see that considering nodes only in $\mathcal{N}_i\cap\mathcal{N}_j$ is sufficient. The diameter of both $\mathcal{C}$ and $\mathcal{C}'$ is $r$. If $v_k$ lies in one of the circles, the distance between $v_k$ and either of $v_i$ and $v_j$ is less than $r$, and hence, $v_k$ is a neighbor to both. \\

We now derive the following:
\begin{enumerate}
\item A test to see if a node lies in both  circles $\mathcal{C}$ and $\mathcal{C}'$. This immediately determines that $e$ does not belong to $\partial A_r$.
\item A test to see if a node lies in exactly one of the circles $\mathcal{C}$ and $\mathcal{C}'$.
\item Given that there exists at least one node in one of the circles, a test to see if a subsequent node lies in the other. This also, immediately determines that $e$ does not belong to $\partial A$.
\end{enumerate}

Let the angle subtended by the chord $v_iv_j$  on the bigger arc (of either circle, see Figure \ref{fig:cycleConstruction}) be $\theta$, hence making the angle subtended on the smaller arc  $\pi - \theta$.  The angle which the chord subtends at the center, $\omega$, may easily be computed using the law of cosines, and $\theta$ is equal to $\omega/2$.\\

Let $v_k \in \mathcal{N}_i \cap \mathcal{N}_j$, and  let $\angle{v_iv_kv_j}=\phi_k$. Then, if $\phi_k > \pi-\theta$, $v_k$ lies inside both  circles, and we immediately know that $e$ does not belong to $\partial A_r$.  If $\phi \le \theta$, it lies inside neither circle. Let $v_k$ be the first node satisfying $\theta < \phi \le \pi - \theta$. Then $v_k$ lies in one of the circle. Without loss of generality, we assume $v_k$ lies in $\mathcal{C}$. \\

Let $v_l$ be any subsequent node satisfying $\theta < \phi_l \le \pi - \theta$. If $v_l$ is not a neighbor of $v_k$, then $v_l$ lies in $\mathcal{C}'$, since any two nodes inside a circle of  diameter $r$ will be neighbors. If $v_k$ and $v_l$ are neighbors, we know the length $\|(v_kv_l)\|$. Using the law of cosines, we compute the angle $\angle{v_kv_iv_l}$ which we call $\beta$. If $\beta = \angle{v_kv_iv_j}+\angle{v_lv_iv_j}$, $v_l$ lies in $\mathcal{C}'$, and if $\beta = |\angle{v_kv_iv_j}-\angle{v_lv_iv_j}|$, $v_l$ lies in $\mathcal{C}$. Figure \ref{fig:angleRelationshipsForCycleSides} demonstrates this relationship between the angles.\\

\begin{table}[!ht]
\centering
\begin{tabular}{l}
\hline \\
computing the $\alpha-$shape\\
\hline\\
At each edge $e=(v_i,v_j)$ in $G$,  \\
\hspace{1cm} compute $\theta$ \\
\hspace{1cm} for each $v_k \in \mathcal{N}_i\cap \mathcal{N}_j$   \\
\hspace{1.5cm} compute $\phi_k$ \\
\hspace{1.5cm}if $\phi_k > \pi - \theta$ \\
\hspace{2cm}$e \not\in \partial A$, terminate \\
\hspace{1.5cm}if $\phi_k \le \theta$, \\
\hspace{2cm}continue to next node\\
\hspace{1.5cm}if $\theta < \phi \le \pi - \theta $ \\
\hspace{2cm}is $v_k$ the first node satisfying this condition? \\
\hspace{2.2cm}assign $v_k$ to $\mathcal{C}$ \\
\hspace{2cm}else  \\
\hspace{2.2cm}compute $\beta$ \\
\hspace{2.2cm}if $\beta = |\angle{v_kv_iv_j}-\angle{v_lv_iv_j}|$ \\
\hspace{2.4cm}continue to next node \\
\hspace{2.2cm}else\\
\hspace{2.4cm}$e\not\in \partial A$, terminate \\
\hspace{1cm} $e \in \partial A$ \\
\hline
\end{tabular}
\caption{Algorithm for computing the $\alpha-$shape. Note that all the computations require only local information.}
\label{Tab:AlgoForAlphaShape}
\end{table}

The algorithm terminates when we determine that both  circles $\mathcal{C}$ and $\mathcal{C}'$ contain at least one node, or there are no more nodes in $\mathcal{N}_i \cap \mathcal{N}_j$ to consider. In the former case, the edge $e$ does not belong to $\partial A_r$ and in the latter, $e$ belongs to $\partial A_r$. Clearly, we can use the same algorithm to compute the $\alpha-$shape for any parameter $0<q \le r$.  The algorithm is summarized in Table \ref{Tab:AlgoForAlphaShape}. Figure \ref{fig:alpha_complex_computation} shows   $\partial A_r$  for a set of points in $\mathbb{R}^2$ computed using the algorithm in Table \ref{Tab:AlgoForAlphaShape}.   The shaded region in the Figure is the union of balls of radius $r/2$ centered at each point. Note that the $\alpha-$shape is a boundary of an object which is homotopy equivalent to the shaded region.

\begin{figure}[!ht]
\centering
\subfigure[]{
\begin{tikzpicture}[scale=.8]
\draw[very thick,red] circle(2) (0,2) circle(2);
\foreach \x in {0.1}
\draw[fill] (0,0) circle(\x) (1.732,1) node[anchor=west]{$v_j$} circle(\x) (-1.732,1) circle(\x) node[anchor=east]{$v_i$};
\draw[fill] (-1,3) circle(0.1) node[anchor=south]{$v_k$};
\draw (-1.732,1) --  (1.732,1) -- (0,-2) -- (-1.732,1);
\draw (-1.732,1) -- (0,2) -- (1.732,1);
\draw[dashed] (-1.732,1) -- (-1,3) -- (1.732,1);
\begin{scope} \clip (-1.732,1)--(1.732,1)--(0,-2)--(-1.732,1); \draw (0,-2) circle (0.3); \end{scope}
\begin{scope} \clip (-1.732,1)--(0,2)--(1.732,1)--(-1.732,1); \draw (0,2) circle (0.3); \end{scope}
\begin{scope} \clip (-1.732,1)--(-1,3)--(1.732,1)--(-1.732,1); \draw (-1,3) circle (0.3); \end{scope}
\draw (0,-1.7) node[anchor=south]{$\theta$} (0,1.7) node[anchor=north]{$\pi - \theta$} (-1,3) node[anchor=north west]{$\phi_k$} ;
\draw (2,0) node[anchor=west]{$\mathcal{C}'$} (2,2) node[anchor=west]{$\mathcal{C}$};
\end{tikzpicture}
\label{fig:cycleConstruction} }
\qquad \qquad \qquad
\subfigure[]{
\begin{tikzpicture}[scale=0.8]
\draw[fill] circle(0.1) (2,0) circle(0.1) (0.6,2) circle(0.1) (1.4,2) circle(0.1) (1.4,-2) circle(0.1);
\draw[thick] (0,0) -- (2,0) -- (0.6,2) -- (0,0) -- (1.4,2) -- (2,0);
\draw[thick,dashed] (0,0) -- (1.4,-2) -- (2,0);
\begin{scope} \clip (1.4,2)--(0,0)--(0.6,2); \draw (0,0) circle (0.4); \end{scope}
\begin{scope} \clip (0.6,2)--(0,0)--(1.4,-2); \draw (0,0) circle (0.5); \end{scope}
\draw (0,0) node[anchor=east]{$v_i$} (2,0) node[anchor=west]{$v_j$} (0.6,2) node[anchor=south]{$v_k$} (1.4,2) node[anchor=south]{$v_l$} (1.4,-2) node[anchor=north]{$v_l$};
\draw (0.15,0.4) node[anchor=south west]{$\beta$} (0.4,-0.2) node[anchor=north west]{$\beta$};
\end{tikzpicture}
\label{fig:angleRelationshipsForCycleSides}
}
\caption{(a)shows the cycles $\mathcal{C}$ and $\mathcal{C}'$ which pass through $v_i$ and $v_j$. $\phi_k$ satisfies $\theta \le \phi \le \pi - \theta$ and $v_k$ lies in $\mathcal{C}$. (b)shows the angle relationships. When $v_l$ also lies in $\mathcal{C}$, then $\beta = |\angle{v_kv_iv_j}-\angle{v_lv_iv_j}|$, and if $v_l$ lies in $\mathcal{C}'$, $\beta = \angle{v_kv_iv_j}+\angle{v_lv_iv_j}$ }
\end{figure}
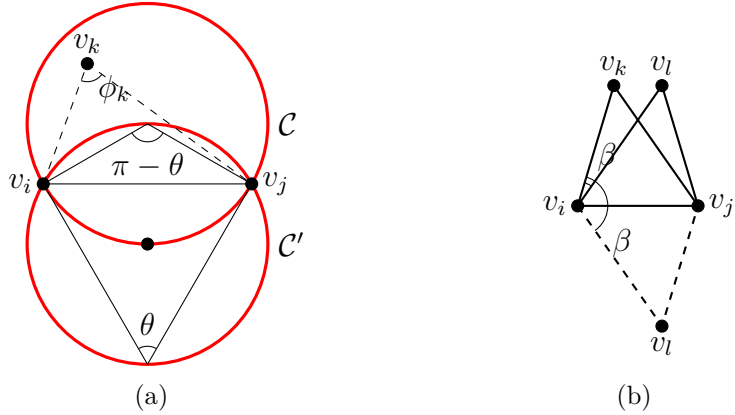

\begin{figure}[!ht]
\centering
\includegraphics[width=0.4\textwidth]{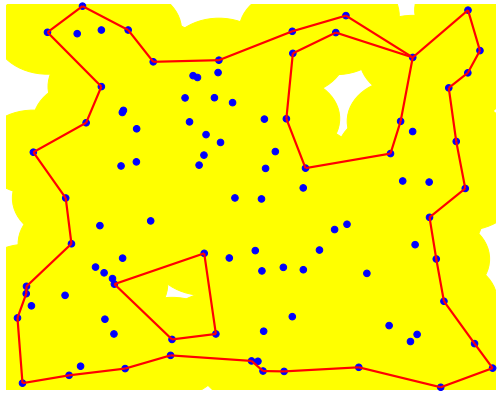}
\caption{$\alpha-$shape with parameter $r_c/2$ for a set of points in $\mathbb{R}^2$ computed using algorithm in Table \ref{Tab:AlgoForAlphaShape}. The shaded region is the union of balls of radius $r_c/2$ centered at each point.}
\label{fig:alpha_complex_computation}
\end{figure}

\section{Relation between $D\check{C}_r$ and  $A_r$}
\label{Relation between RDT(G) and  alpha- complex}

Consider the Delaunay-\v{C}ech complex $D\check{C}_r$ as defined in Section \ref{sec:Restricted Delaunay Triangulation}, and the $\alpha-$complex $A_r$, as defined in Section \ref{sec: alpha-complex and alpha-shape}. We will show that $D\check{C}_r$ has the same homotopy type as $A_r$, by showing that there exists a bijective pairing between the 1-simplices and 2-simplices in $D\check{C}_r\setminus A_r$, such that the pairing describes a homotopy collapse. Note that both $A_r$ and $D\check{C}_r$ do not contain any simplices of dimension greater than 2. Figure \ref{fig:homotopy collapse} shows such a homotopy collapse. \\

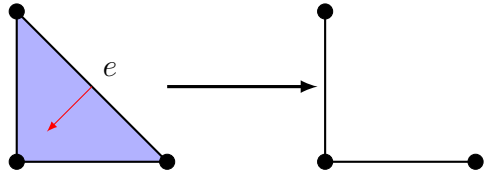
\begin{figure}[!ht]
\centering
\begin{tikzpicture}[>=latex]
\draw[fill, blue!30!white] (0,0) -- (2,0) -- (0,2) -- (0,0);
\draw[thick] (0,0) -- (2,0) -- (0,2) -- (0,0);
\draw[fill] circle(0.1) (2,0) circle(0.1) (0,2) circle(0.1);
\draw[red,->] (1,1) -- (0.4, 0.4);
\draw (1,1) node[anchor=south west]{$e$};
\draw[->,very thick] (2,1) -- (4,1);
\draw[thick,xshift=4.1cm] (0,0) -- (2,0); \draw[thick,xshift=4.1cm] (0,2) -- (0,0);
\draw[fill,xshift=4.1cm] circle(0.1) (2,0) circle(0.1) (0,2) circle(0.1);
\end{tikzpicture}
\caption{Homotopy collapse of an edge into an adjacent 2-simplex}
\label{fig:homotopy collapse}
\end{figure}

Let $F(G_r)$ be the flag complex of $G_r$. Define the complex $U_r$ as $U_r = DT(V) \cap F(G_r) $. Since the \v{C}ech complex $\check{C}(V,r)$ is a subcomplex of $F(G_r)$, $D\check{C}_r$ is a subcomplex of $U_r$. \\

Let $\mathcal{T}_e$ denote the set of all 2-simplices to which $e$ is a face in $U_r$, $\mathcal{T}_e^{r_c/2} \subseteq \mathcal{T}_e$ denote the 2-simplices in $\mathcal{T}_e$ with circum-radius less than or equal to $r/2$,  and $\mathcal{T}_e^{\pi/2} \subseteq \mathcal{T}_e$  denote the 2-simplices in $\mathcal{T}_e$ such that the angle opposite $e$ is greater than $\pi/2$. Figure \ref{fig:examples of triangle sets} shows examples of triangles with these properties.\\

\begin{figure}
\centering
\subfigure[$\tau \in \mathcal{T}_e^{r_c/2}$]{
\begin{tikzpicture}
\draw[fill,yellow!30!white] (0,0) circle(1.2) (2,0) circle(1.2) (1,1.732) circle(1.2);
\draw[yellow] (0,0) circle(1.2) (2,0) circle(1.2) (1,1.732) circle(1.2);
\draw[thick] (0,0) -- (2,0) -- (1,1.732) -- (0,0);
\draw[fill] (0,0) circle(.1) (2,0) circle(.1) (1,1.732) circle(.1);
\draw (1,0) node[anchor=north]{$e$} (1.5,0)node[anchor=south]{$\tau$};
\draw (1,1.732) -- (1,2.932); \draw (1,2.332) node[anchor=west]{$r_c/2$};
\end{tikzpicture}
\label{fig:triangle with small radius}
}
\hspace{1cm}
\subfigure[$\tau \in \mathcal{T}_e^{\pi/2}$]{
\begin{tikzpicture}
\draw[fill, blue!30!white] (0,0) -- (2,0) -- (-1,2) -- (0,0);
\draw[thick] (0,0) -- (2,0) -- (-1,2) -- (0,0);
\draw[fill] circle(0.1) (2,0) circle(0.1) (-1,2) circle(0.1);
\draw (0.6,1) node[anchor=south west]{$e$};
\draw (0,0.8) node{$\tau$};
\begin{scope} \clip (0,0) -- (2,0) -- (-1,2) -- (0,0); \draw circle(0.2); \end{scope}
\draw node[anchor=south west]{$>\pi/2$};
\end{tikzpicture}
\label{fig:obtuse angled triangle}
}
\caption{Examples of the triangle sets $\mathcal{T}_e^{r_c/2}$ and $\mathcal{T}_e^{\pi/2}$. The circles in (a) have are centered at the nodes with radius $r/2$. Since the circum-radius is less than $r/2$, the circles have a common intersection with positive area.}
\label{fig:examples of triangle sets}
\end{figure}
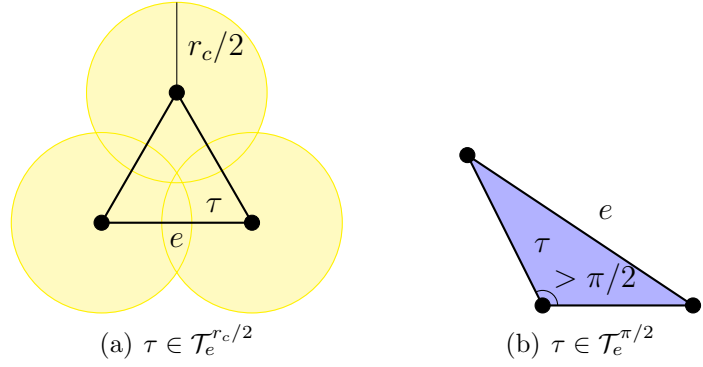

In order to show  the existence of a paring, we first analyze the triangles surrounding an edge $e \in A_r)$. The following lemma characterizes the 2-simplices in $A_r$ in terms of their circum-radius.

\begin{lemma}
\label{lem:triangleInAlpha}
A simplex $(v_1,v_2,v_3)$ in $U_r$ is in $A_r$ iff the circum-radius of the triangle $(v_1,v_2,v_3)$ \footnote[2]{we use the notation $(v_1,v_2,v_3)$ to denote both the simplex and the underlying triangle} is less than or equal to $r/2$.
\end{lemma}

\begin{proof}
Since $U_r \subseteq DT(V)$, $(v_1,v_2,v_3) \in U_r \Rightarrow (v_1,v_2,v_3) \in DT(V) $. The circumradius is less than or equal to $r/2$, \emph{iff} circumcenter belongs to all  $\alpha-$cells $\alpha(v_1)$, $\alpha(v_2)$ and $\alpha(v_3)$. This results in the three $\alpha-$ cells having a non-empty intersection, hence $(v_1,v_2,v_3) \in A_r$. $\blacksquare$
\end{proof}

Lemmas \ref{lem:midpointWitness} and \ref{lem:conditionsForEdgeinAlpha} together impose conditions on cardinality of the sets $\mathcal{T}_e^{r_c/2}$ and $\mathcal{T}_e^{\pi/2}$. We utilize these conditions in Lemma \ref{theo:pairing theorem} to show the existence of the pairing.

\begin{lemma}
\label{lem:midpointWitness}
Denote by $m_e$ the midpoint of the 1-simplex $e \in U_r$. $m_e$ is a witness for $e$ iff  $\mathcal{T}_e^{\pi/2} = \varnothing$
\end{lemma}

\begin{proof}
$m_e$ is a witness for $e$ iff there does not exist any other node inside the circle with $e$ as the diameter (illustrated in Figure \ref{fig:midpointFigure}). This occurs if and only if the angle opposite $e$ in any incident triangle  is acute. $\blacksquare$
\end{proof}

\begin{figure}[!th]
\centering
\subfigure[Construction for Lemma \ref{lem:midpointWitness}.]{
\begin{tikzpicture}
\draw[very thick,red] circle(2);
\draw[fill] (-2,0) circle(0.1) (2,0) circle(0.1) (0,1) circle(0.1) (0.6, 3) circle(0.1) (0,0) circle(0.1);
\draw (-2,0) -- (2,0) -- (0,1) -- (-2,0); \draw (2,0) -- (0.6,3) -- (-2,0);
\draw (0,1) node[anchor=north]{$> \pi/2$} (0.6,3) node[anchor=north west]{$< \pi/2$} (0,0) node[anchor=north]{$m_e$} ;
\end{tikzpicture}
\label{fig:midpointFigure}
}
\hspace{2cm}
\subfigure[Construction for Lemma \ref{lem:onlyOneObtuseTriangle}]{
\begin{tikzpicture}
\draw[very thick,red] circle(2);
\draw[fill] (-1.732,-1) circle(.1) (1.732,-1) circle(.1) (0,-2) circle(.1) (0,2.5) circle(.1) ;
\draw (-1.732,-1) -- (0,-2) -- (1.732,-1) circle(.1)  -- (-1.732,-1) circle(.1) -- (0,2.5) -- (1.732,-1) circle(.1) ;
\draw (0,-2) node[anchor=south]{$\phi$}  (0,2.5) node[anchor=north west]{$< \pi - \phi$};
\draw (0,-1) node[anchor=south]{$e$} (0,2.5) node[anchor=south]{$v$} (0,0) node{$\tau_1$} (0,-1.2)node[anchor=north east]{$\tau$};
\end{tikzpicture}
\label{fig:onlyOneTriangle}
}
\caption{}
\end{figure}

\begin{lemma}
\label{lem:conditionsForEdgeinAlpha}
 Consider the following statements
\begin{itemize}
\item $S_1:$ $\mathcal{T}_e^{\pi/2} = \varnothing$
\item $S_2:$ $\mathcal{T}_e^{r/2} \ne \varnothing$
\item $S_3:$ $e \in A_r$
\end{itemize}

$S_1 \vee S_2 $ is a necessary and sufficient condition for $S_3$.

\end{lemma}

\begin{proof}
Let $e=(v_1,v_2)$. For sufficiency: from Lemma \ref{lem:midpointWitness}, $S_1$ implies $m_e$ is a witness for $v_1$ and $v_2$, and $S_2$ implies $\exists$ a witness (the circumcenter of one of the triangles in $\mathcal{T}_e^{r/2}$). For necessity: if $S_3$ is true, then there exists a witness for $e$. If $m_e$ is a witness, then $S_1$ is true. If $m_e$ is not a witness, then $e$ shares a witness with a 2-simplex which is in $A_r$. From Lemma \ref{lem:triangleInAlpha}, this implies $\mathcal{T}_e^{r/2} \ne \varnothing $. Therefore, $S_2$ is true. $\blacksquare$
\end{proof}

The above Lemma suggests the existence or non-existence of types of triangles surrounding an edge in $A_r$. Lemma \ref{lem:onlyOneObtuseTriangle} and Theorem \ref{theo:pairing theorem} further refine this relationship, and precisely identify the triangle to be removed when an edge is removed from $D\check{C}_r$.

\begin{lemma}
\label{lem:onlyOneObtuseTriangle}
If $\mathcal{T}_e^{\pi/2} \ne \varnothing$, then $|\mathcal{T}_e^{\pi/2}| = 1$.
\end{lemma}

\begin{proof}
Suppose $\mathcal{T}_e^{\pi/2} \ne \varnothing$. Let $\tau \in \mathcal{T}_e^{\pi/2}$, and let the angle opposite $e$ in $\tau$ be $\phi$ with $\phi>\pi/2$ (see Figure \ref{fig:onlyOneTriangle}). Let $\tau_1 \ne \tau$ be incident on $e$, with $v$ being the opposite vertex. Since $\tau \in U_r$, $v$ does not lie inside the circum-center of $\tau$. This implies that the angle opposite $e$ in $\tau_1$ is less than $\pi - \phi$ which is less than $\pi/2$. $\blacksquare$
\end{proof}

Let $C_k(K)$ denote the $k-$simplices in the complex $K$.

\begin{theorem}
\label{theo:pairing theorem}
Let $\mathcal{T}_R = C_2\left(D\check{C}_r\right)\setminus  C_2\left(A_r\right)$ and $\mathcal{E}_R = C_1\left(D\check{C}_r\right) \setminus C_1\left(A_r\right)$. There exists a bijective pairing $P: \mathcal{E}_R \rightarrow \mathcal{T}_R$ such that $e$ is a face of $P(e)$.
\end{theorem}

\begin{proof}
Let $e\in D\check{C}_r$ but $e \not\in A_r$, from  Lemma \ref{lem:conditionsForEdgeinAlpha}, $\mathcal{T}_e^{r_c/2}=\varnothing$ and $\mathcal{T}_e^{\pi/2} \ne \varnothing$. Since we assume $e\not\in A_r$, $e$ cannot be a face of any 2-simplex in $A_r$. Owing to the condition $\mathcal{T}_e^{r_c/2}=\varnothing$, Lemma \ref{lem:triangleInAlpha} ensures that this is indeed the case.    Also, from Lemma \ref{lem:onlyOneObtuseTriangle}, $|\mathcal{T}_e^{\pi/2}| = 1$. Let $\tau \in \mathcal{T}_e^{\pi/2}$. Note that $\tau$ is unique, and $\tau \not\in A_r$. Further, since $\tau$ is an obtuse triangle, $\tau \in \check{C}(V,r)$, and this implies $ \tau \in D\check{C}_r$. The pairing $P$ is then defined as $P(e) = \tau$. $\blacksquare$
\end{proof}

%

For any simplicial complex $K$, let $\sigma_1$ and $\sigma_2$ be simplices of dimension $1$ and $2$ such that $\sigma_1$ is a face of $\sigma_2$. Then, there exists a deformation retraction $F_{\sigma_1}: K \rightarrow K \setminus (\sigma_1 \cup Int(\sigma_2))$, which ``collapses'' $\sigma_1$ into $\sigma_{2}$. Therefore, $K$ is homotopy equivalent to $K \setminus (\sigma_1 \cup Int(\sigma_{2}))$.\\

The removal of edges $\mathcal{E}_R$ and triangles $\mathcal{T}_R$ describes a finite sequence of deformation retractions via the pairing $P$. When we collapse all the edges into their paired triangles, the resulting complex is $A_r$. Each collapse is a homotopy equivalence, and a composition of homotopy equivalences is a homotopy equivalence. This leads us to our main theorem:

\begin{theorem}
The complexes $D\check{C}_r$ and $A_r$ are homotopy equivalent.
\end{theorem}

Figure \ref{fig:homotopyEquivalence} illustrates the above theorem using an example. Note that   $A_r$ and  $D\check{C}_r$ are homotopy equivalent to each other and both are homotopy equivalent to $R_c$ (the shaded region). Further, as seen, $D\check{C}_r$ is a better geometric approximation to $R_c$ than $A_r$. This is simply because $A_r$ is a sub-complex of $D\check{C}_r$.

\begin{figure}
\centering
\subfigure[$A_{r_c/2}(V)$]{
\includegraphics[width=0.3\textwidth]{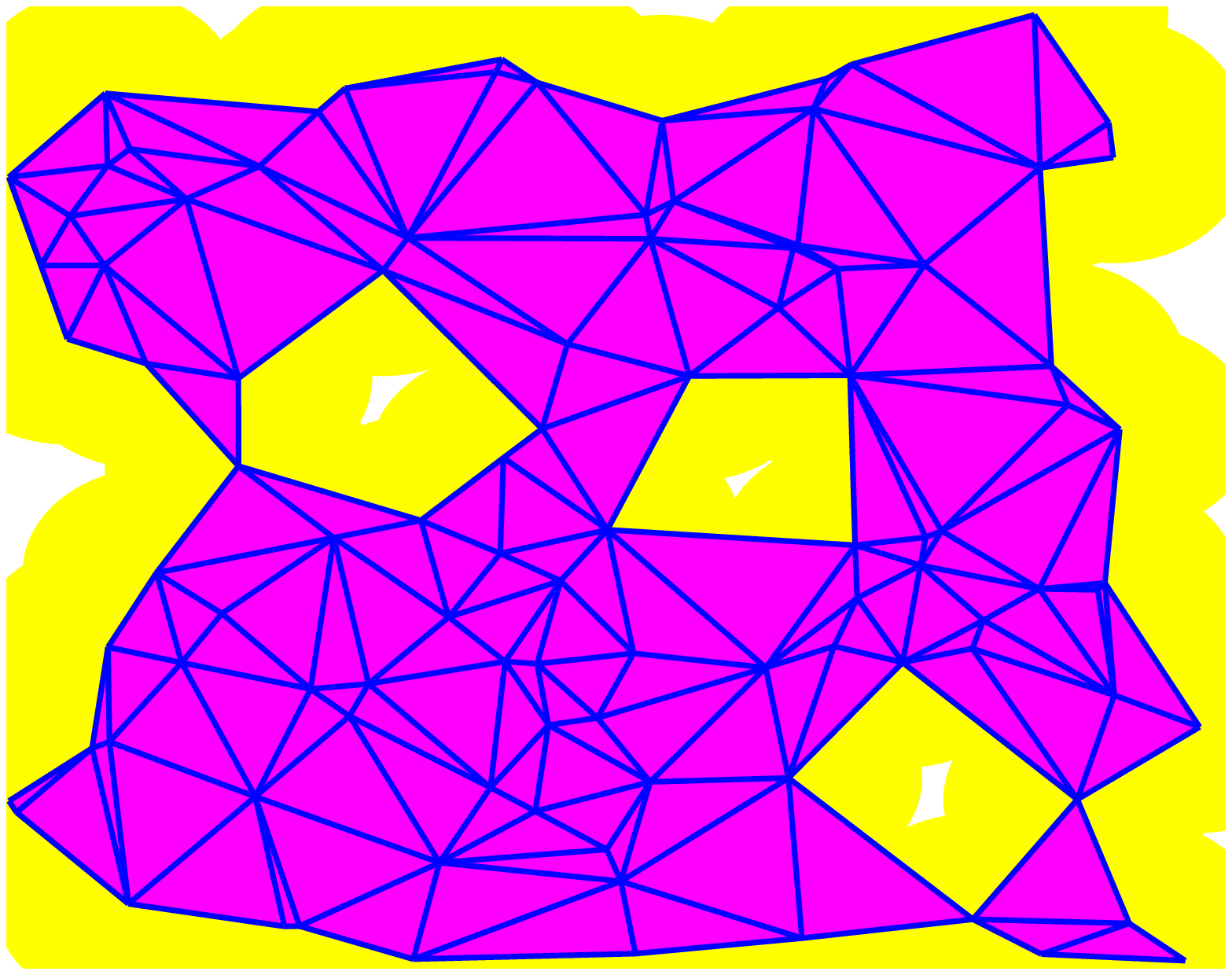}
}
\subfigure[$D\check{C}_{r_c/2}(V)$]{
\includegraphics[width=0.3\textwidth]{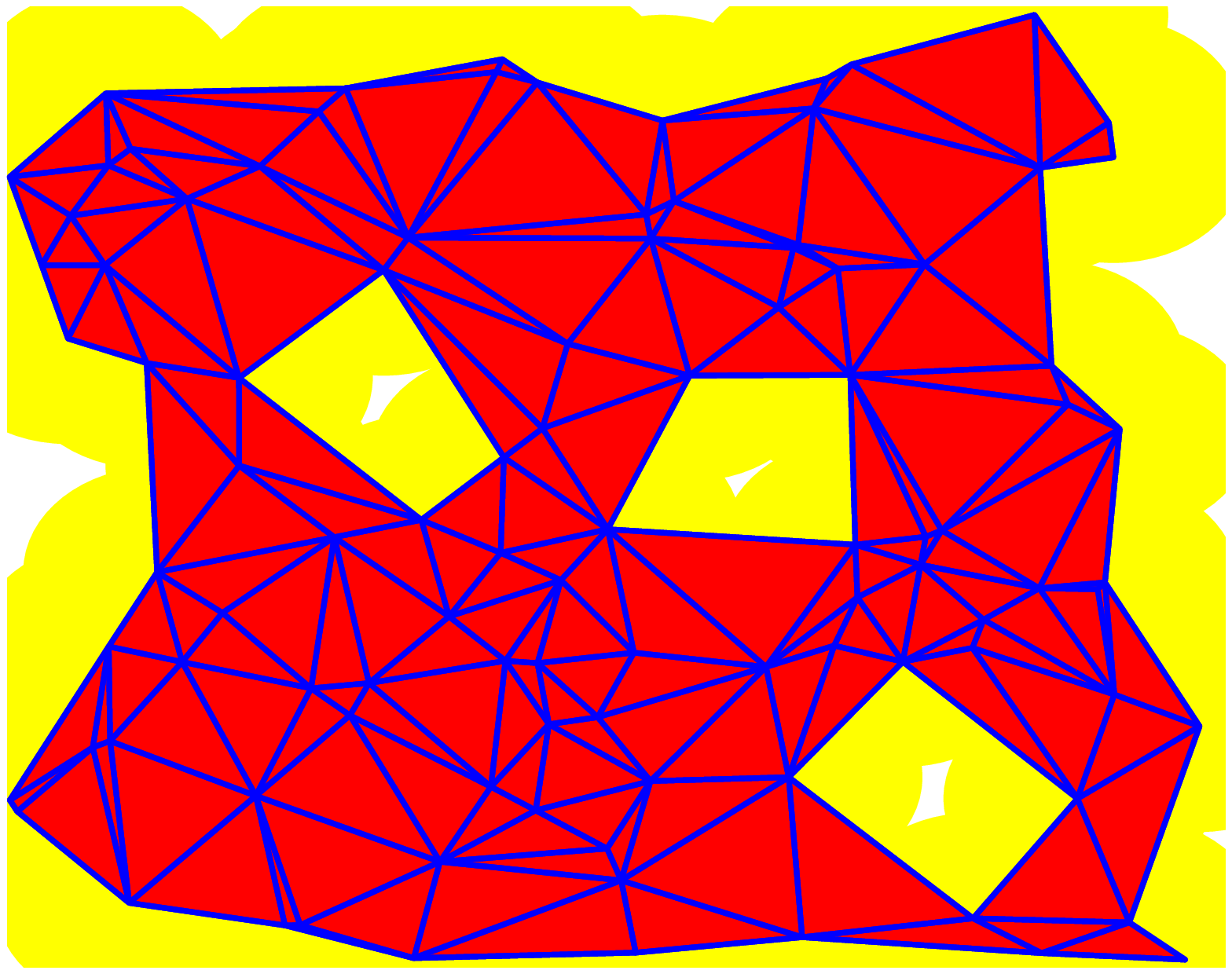}
}
\subfigure[$A_{r_c/2}(V)$ super-imposed over $D\check{C}_{r_c/2}(V)$]{
\includegraphics[width=0.3\textwidth]{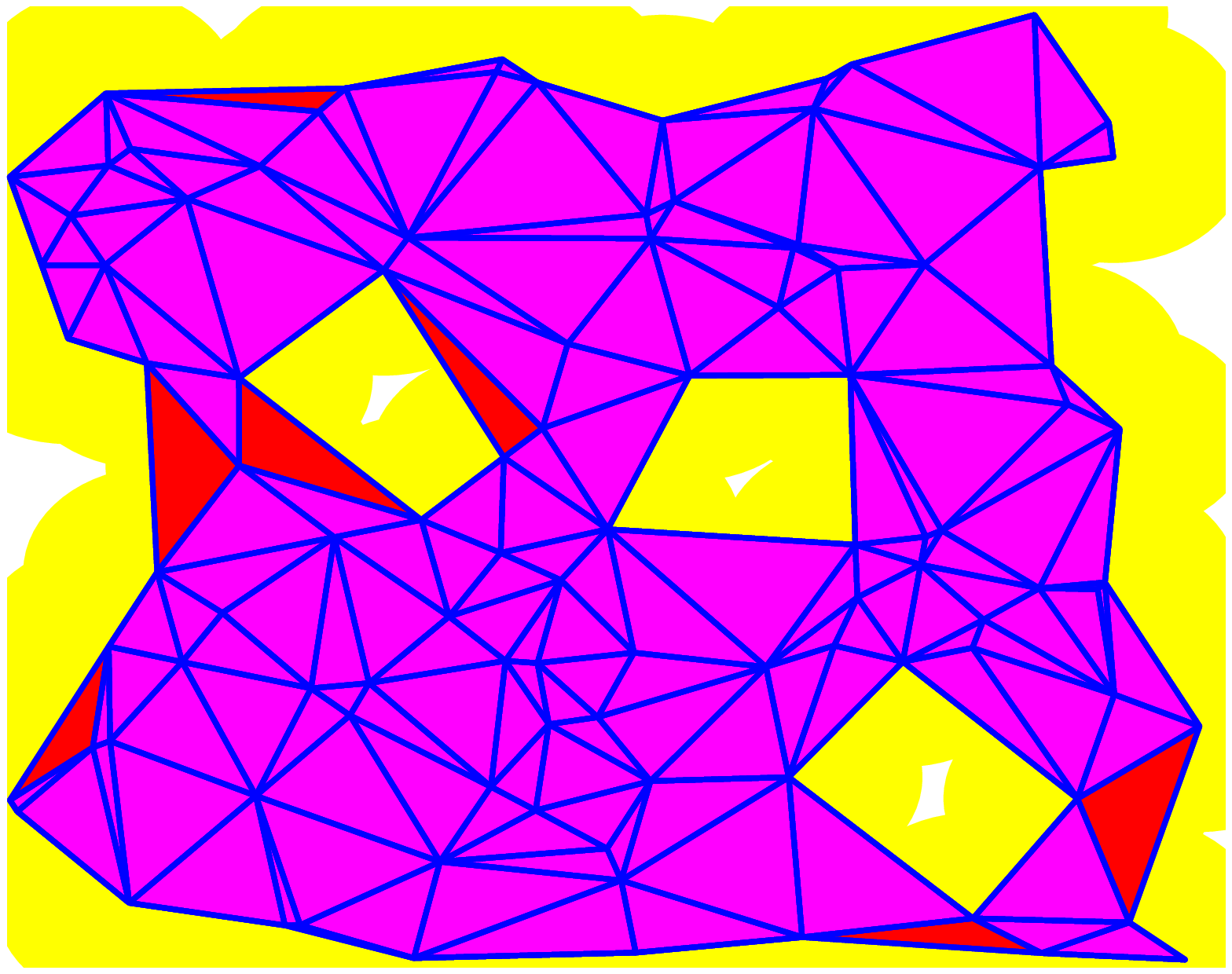}
}
\caption{Figure shows the homotopy equivalence between $A_{r_c/2}(V)$ and $D\check{C}_{r_c/2}(V)$. The shaded region is $R_c$. Note that $D\check{C}_{r_c/2}(V)$ is a better geometric approximation to $R_c$ than $A_{r_c/2}(V)$. }
\label{fig:homotopyEquivalence}
\end{figure}


\section{Conclusion}
\label{sec:Conclusion}
The algorithm described in Section \ref{sec:computingAlphaShape} takes the edge lengths as inputs and outputs the alpha shapes. We make no further assumptions on the node density, and we need not compute any  coordinates. The decision about an edge belonging to an $\alpha-$shape is carried out by only looking at the local information, i.e., considering only the points within a certain distance, and may therefore be implemented distributively. In Section \ref{sec:Restricted Delaunay Triangulation}, we define the Delaunay-\v{C}ech complex which contains the alpha complex. Its boundary, defined as a Delaunay-\v{C}ech shape, is therefore a better geometric approximation for the union of balls with an appropriate radius. We also show in Section \ref{Relation between RDT(G) and  alpha- complex} that, like the $\alpha-$shape, the Delaunay-\v{C}ech shape remains topologically faithful to the underlying space.\\

\bibliographystyle{plain}
\bibliography{theBigBib}

\end{document}